\documentclass[11pt]{article}

\usepackage[letterpaper,margin=0.90in]{geometry}
\usepackage{amsmath, amssymb, amsthm, thmtools, amsfonts}
\usepackage{bbm}

\usepackage{ifthen}

\usepackage{tikz}
\usetikzlibrary{positioning,decorations.pathreplacing}

\usepackage{cite}
\usepackage{appendix}
\usepackage{graphicx}
\usepackage{color}
\usepackage{algorithm}
\usepackage[noend]{algpseudocode}
\usepackage{epstopdf}
\usepackage[textsize=tiny]{todonotes}

\usepackage{framed}
\usepackage[framemethod=tikz]{mdframed}
\usepackage[bottom]{footmisc}
\usepackage[shortlabels]{enumitem}
\setitemize{noitemsep,topsep=3pt,parsep=3pt,partopsep=3pt}
\usepackage[font=small]{caption}
\usepackage{xspace}

\newtheorem{theorem}{Theorem}[section]
\newtheorem{lemma}[theorem]{Lemma}
\newtheorem{meta-theorem}[theorem]{Meta-Theorem}

\newtheorem{remark}[theorem]{Remark}

\definecolor{darkgreen}{rgb}{0,0.5,0}
\usepackage{hyperref}
\hypersetup{
    unicode=false,          
    colorlinks=true,        
    linkcolor=red,          
    citecolor=darkgreen,        
    filecolor=magenta,      
    urlcolor=cyan           
}
\usepackage[capitalize, nameinlink]{cleveref}
\crefname{theorem}{Theorem}{Theorems}
\Crefname{lemma}{Lemma}{Lemmas}
\Crefname{figure}{Figure}{Figures}
\algnewcommand\algorithmicswitch{\textbf{switch}}
\algnewcommand\algorithmiccase{\textbf{case}}

\algdef{SE}[SWITCH]{Switch}{EndSwitch}[1]{\algorithmicswitch\ #1\ \algorithmicdo}{\algorithmicend\ \algorithmicswitch}%
\algdef{SE}[CASE]{Case}{EndCase}[1]{\algorithmiccase\ #1}{\algorithmicend\ \algorithmiccase}%
\algtext*{EndSwitch}%
\algtext*{EndCase}%

\newcommand{\eps}{\varepsilon}

\newcommand{\poly}{\operatorname{\text{{\rm poly}}}}

\renewcommand{\paragraph}[1]{\vspace{0.15cm}\noindent {\bf #1}:}

\newcommand{\FullOrShort}{full}

\ifthenelse{\equal{\FullOrShort}{full}}{
	
  \newcommand{\fullOnly}[1]{#1}
  \newcommand{\shortOnly}[1]{}
  
  }{

    \newcommand{\fullOnly}[1]{}
    \newcommand{\IncludePictures}[1]{}
   
  }

\begin{document}

\date{}

\title{A Simple Parallel and Distributed Sampling Technique: \\ Local Glauber Dynamics}

\author{
	 Manuela Fischer\\
   ETH Zurich \\
   manuela.fischer@inf.ethz.ch
	\and
	 Mohsen Ghaffari\\
   ETH Zurich \\
   ghaffari@inf.ethz.ch
 }

\maketitle

\setcounter{page}{0}
\thispagestyle{empty}

\begin{abstract}
\emph{Sampling}
 constitutes an important tool in a variety of areas: from machine learning and combinatorial optimization to computational physics and biology.  
A central class of sampling algorithms is the \emph{Markov Chain Monte Carlo} method, based on the construction of a Markov chain with the desired sampling distribution as its stationary distribution. Many of the traditional Markov chains, such as the \emph{Glauber dynamics}, do not scale well with increasing dimension.
To address this shortcoming, we propose a simple local update rule based on the Glauber dynamics that leads to efficient parallel and distributed algorithms for sampling from Gibbs distributions. 
\medskip

Concretely, we present a Markov chain that mixes in $O(\log n)$ rounds when Dobrushin's condition for the Gibbs distribution 
 is satisfied. 
This improves over the \emph{LubyGlauber} algorithm by Feng, Sun, and Yin [PODC'17], which needs $O(\Delta \log n)$ rounds, and their \emph{LocalMetropolis} algorithm, which converges in $O(\log n)$ rounds but requires a considerably stronger mixing condition. Here, $n$ denotes the number of nodes in the 
graphical model inducing the Gibbs distribution, and $\Delta$ its maximum degree. In particular, our method can sample a uniform proper coloring with $\alpha\Delta$ colors in $O(\log n)$ rounds for any $\alpha >2$, which almost matches the threshold of the sequential Glauber dynamics and improves on the $\alpha>2 + \sqrt{2}$ threshold of Feng et al. 
\end{abstract}

\newpage

\section{Introduction}

\emph{Locally Checkable Labeling} (LCL) \cite{naor1995localicty} problems have been studied extensively for more than three decades \cite{linial1987LOCAL}. Sampling from the solution space of such LCLs, however, has not attracted a lot of attention and has been investigated only by a recent work \cite{feng2017can}, despite its numerous motivations, which we will outline in the following. 

\paragraph{Markov Chain Monte Carlo Method} The \emph{Markov Chain Monte Carlo (MCMC)} method is a central class of algorithms for \emph{sampling}, that is, for randomly drawing an element from a ground set according to a certain probability distribution. It works by constructing a Markov chain with the targeted sampling distribution as its stationary distribution. Within a number of steps, known as the mixing time, the Markov chain converges; its state then (approximately) follows this distribution.
Besides the intrinsic interest of such a general sampling method, in particular for complex distributions where simple sampling techniques fail, the MCMC method gives rise to efficient approximation algorithms in a variety of areas: enumerative combinatorics (due to the fundamental connection between sampling and counting established by Jerrum, Valiant, and Vazirani \cite{jerrum1986random}), simulated annealing \cite{nahar1986simulated} in combinatorial optimization, Monte Carlo simulations \cite{metropolis1953equation} in statistical physics, computation of intractable integrals for, among many others, Bayesian inference \cite{andrieu2003introduction} in machine learning.

\paragraph{Parallel and Distributed Sampling}
The employment of MCMC methods is particularly important when confronted with high-dimensional data, where traditional (exact) approaches quickly become intractable.   
Such data sets are not only increasingly frequent, but also critical for the success of many applications. For instance in machine learning, higher-dimensional models help expressability and hence predictability. 
It is thus central that MCMC algorithms scale well with increasing dimensions. 
This is not the case, however, for most sequential methods, as they process and update the variables one by one, that is, a single site per step. 
To speed up the sampling process, Markov chain updates can be parallelized by spreading the variables across several processors. 
In other settings, such as distributed machine learning, the (data associated to) variables might already be naturally distributed among several machines, and the overhead of aggregating them into one machine, if they fit there in the first place, would be untenable. 

\paragraph{Local Sampling} In either case, to avoid overhead in communication and coordination, 
local update rules for Markov chains are needed: a machine must be able to change the value of its variables without knowing all the values of the variables on other machines. Yet, the joint distribution, over all variables in the system, must converge to a certain distribution. This local sampling problem was introduced in a recent work by Feng, Sun, and Yin \cite{feng2017can}, whose title asks 
\emph{``What can be sampled locally?''}. 
We address this question by providing a simple and generic sampling technique---the \emph{Local Glauber Dynamics}, informally introduced in \Cref{sec:technique} and formally described in \Cref{sec:LocalGlauber}---which is applicable for a wide range of distributions, as stated in \Cref{sec:results}. This moves us a step closer towards an answer of this question, thus towards the goal of generally understanding what can be sampled locally. Besides its many practical ramifications, especially on the area of distributed machine learning, this gives us a theoretical insight about the \emph{locality} of problems, whose systematic study has been initiated by the seminal works of Linial \cite{linial1987LOCAL} and Naor and Stockmeyer \cite{naor1995localicty} with the pithy title \emph{``What can be computed locally?''}. 

\subsection{Our Result, and Related Work}\label[section]{sec:results}

For the sake of succinctness and comprehensibility of the presentation, we state and prove our main result in terms of the special case that gets most attention for sequential sampling: sampling proper colorings of a graph.
We refer to \cite{frieze2007survey} for a survey on sequential sampling of proper colorings. Our result applies to a more general set of distributions, however, as explained in the remark at the end of this section. Note that independently and simultaneously, Feng, Hayes, and Yin \cite{feng2018distributed} arrived at the same result. 

\begin{theorem}\label{thm:col}
A uniform proper $q$-coloring of an $n$-node graph with maximum degree $\Delta$ can be sampled within total variation distance $\eps>0$ in $O\left(\log \left(\frac{n}{\eps}\right)\right)$ rounds, where $q=\alpha\Delta$ for any $\alpha>2$. 
\end{theorem}

Our parallel and distributed sampling algorithm improves over the \emph{LubyGlauber} algorithm by Feng, Sun, and Yin \cite{feng2017can}, which needs $O\left(\Delta \log \left(\frac{n}{\eps}\right)\right)$ rounds, and their \emph{LocalMetropolis} algorithm, which converges in $O\left(\log \left(\frac{n}{\eps}\right)\right)$ rounds but requires a considerably stronger mixing condition of $\alpha > 2+\sqrt{2}$. They state that \emph{``We also believe that the $2+\sqrt{2}$ threshold is of certain significance to this [LocalMetropolis] chain as the Dobrushin's condition to the Glauber dynamics.''}, thus implying that this value is a barrier for their approach. This is also justified by the supposedly easiest special case of a tree that leads to the same threshold for their algorithm. Our result gets rid of the additional $\sqrt{2}$ while not incurring any loss in the round complexity, with a considerably easier and more natural update rule. 
Not only our proof is simpler and shorter, also our algorithm is asymptotically best possible, as there is an $\Omega\left(\log\left(\frac{n}{\eps}\right)\right)$ lower bound \cite{guo2016uniform,feng2017can} due to the exponential correlation between variables. 

The threshold of $\alpha >2 $ corresponds to Dobrushin's condition, thus almost matches the threshold of the sequential Glauber dynamics \cite{jerrum1995very,salas1997absence} at $2\Delta+1$. In other words, we present a technique that fully parallelizes the Glauber dynamics, speeding up the mixing time from $\poly n$ steps to $O(\log n)$ rounds. In terms of number of colors needed, Dobrushin's condition can be undercut: Vigoda \cite{vigoda2000improved} and two very recent works \cite{chen2018linear, delcourt2018rapid} showed that, when resorting to a different highly non-local Markov chain, $\alpha=\frac{11}{6}$ is enough. This gives rise to the question whether efficient distributed algorithms intrinsically need to be stuck at Dobrushin's condition, which would imply that this bound is inherent to the locality of the sampling process, or whether our threshold is an artifact of our possibly suboptimal dynamics. 

\begin{remark}\label[remark]{remark:fullGenerality}
In fact, our technique directly applies for sampling from the Gibbs distribution induced by a Markov random field\footnote{This captures many graph problems---such as independent set, vertex cover, graph homomorphism---and physical models---such as Ising model, Potts model, general spin systems, and hardcore gas model.} if Dobrushin's condition \cite{dobruschin1968description} is satisfied. More generally, it can used for sampling from any local (that is, constant-radius) constraint satisfaction problems, which is universal for conditional independent joint distributions, due to Hammersley-Clifford's fundamental theorem \cite{hammersley1971markov}. Moreover, our proof presented here captures all the difficulties that arise in these more general cases, thus can be adapted in a straight-forward manner. We defer this generalization to the full version of the paper. 
\end{remark}



\subsection{Our Sampling Technique, and Related Approaches}\label[section]{sec:technique}

Over the past few years, several methods to parallelize sequential Markov chains have been proposed. Most of them rely on a heavy coordination machinery, are special purpose, and/or do not provide any theoretical guarantees. In the following, we briefly outline two of the most promising and more generic parallel and distributed sampling techniques, in the context of colorings. 
  
The most natural one follows a standard decentralization approach, also implemented in the \emph{LubyGlauber} algorithm by \cite{feng2017can}: an independent set of nodes (e.g., a color class of a proper coloring) simultaneously updates their color \cite{feng2017can}, ensuring that no two neighboring nodes change their color at the same time. This approach mainly suffers from the limitation that the number of independent sets needed to cover all nodes might be large, which slows down mixing. In particular, a multiplicative $\Delta$-term in the mixing time seems inevitable \cite{gonzalez2011parallel,feng2017can}. In the worst case of a clique, this approach falls back to sequential sampling, updating one node after the other. Moreover, this method requires an independent set to be computed, which incurs a significant amount of additional communication and coordination.

An orthogonal direction was pursued by \cite{newman2008distributed,yan2009parallel,feng2017can}, where methods are introduced to update the colors of all nodes simultaneously. One example is the \emph{LocalMetropolis} algorithm by \cite{feng2017can}. This extreme parallelism, however, comes at a cost of either introducing a bias in the stationary distribution, resulting in a non-uniform coloring \cite{newman2008distributed, yan2009parallel}, or having to demand stronger mixing conditions \cite{feng2017can}.

\paragraph{Our Local Sampling Technique} 
We aim for the middle ground between these two approaches, motivated by the following observation: we do not need to prevent simultaneous updates of adjacent nodes, only simultaneous \emph{conflicting} updates of adjacent  nodes.
Preventing two adjacent nodes in the first place from picking a new color in the same round seems to be way too restrictive, in particular because it is unlikely that both nodes aim for the same new color. On the other hand, if all nodes update their colors simultaneously, a node is expected to have a conflict with at least one of its neighbors, which prevents progress. 

We interpolate between the two extreme cases by introducing a marking probability, so that only a small fraction of a node's neighbors is expected to update the color, and hence also, in worst case, only these can conflict with its update. 
Concretely, we propose the following generic sampling method, which we call  \emph{Local Glauber Dynamics}: In every step, every variable independently marks itself at random with a certain (low) probability. If it is marked, it samples a proposal at random and checks with its neighbors whether the proposal leads to a conflict with their current state or their new proposals (if any). If there is a conflict, the variable rolls back and stays with its current state, otherwise the state is updated. As opposed to sequential sampling, where only one variable per step updates its value, here the expected number of variables simultaneously updating their value is $\Omega(n)$, resulting in the desired speed-up from $O(n\log n)$, say, to $O(\log n)$. Of course, the main technical aspect lies in showing  that this simple update rule converges to the uniform distribution in $O(\log n)$ rounds, which we prove in \Cref{sec:MarkovChain}.

\subsection{Notation and Preliminaries}
\paragraph{Model} We work with the standard distributed message-passing model for the study of \emph{locality}: the $\mathsf{LOCAL}$ model introduced by Linial \cite{linial1987LOCAL}, defined as follows. Given a graph $G=(V,E)$ on $n$ nodes  with maximum degree $\Delta$, the computation proceeds in rounds. In every round, every node can send a message to each of its neighbors. We do not limit the message sizes, but for the algorithm that we present, $O(\log n)$-bit messages suffice. In the end of the computation, every node $v$ outputs a color. The quantity of main interest is the round complexity, i.e., the number of rounds until the joint output of all nodes satisfies a certain condition. We assume that all nodes have knowledge of $\log n$ and $\Delta$.

\paragraph{Markov Chain}
We consider a Markov chain $\textbf{X}=\left(X^{(t)}\right)_{t \geq 0}$, where $X^{(t)}=\left(X^{(t)}_v\right)_{v\in V}\in [q]^V$ is the coloring of the graph in round $t$. We will omit the round index, and use $X=(X_v)_{v\in V}\in [q]^V$ for the coloring at time $t$ and $X'=(X'_v)_{v\in V}\in [q]^V$ for the coloring at time $t+1$, for a $t \geq 0$, instead. 

\paragraph{Mixing Time}
For a Markov chain $\left(X^{(t)}\right)_{t \geq 0}$ with stationary distribution $\mu$, let $\pi_{\sigma}^{(t)}$ denote the distribution of the random coloring $X^{(t)}$ of the chain at time $t\geq 0$, conditioned on $X^{(0)}=\sigma$. The mixing time $\tau_{\text{mix}}(\eps)=\max_{\sigma \in \Omega} \min \left\{ t \geq 0 \colon  d_{\text{TV}}\left(\pi_{\sigma}^{(t)},\mu\right)\right\}$ is defined to be the minimum number of rounds needed so that the Markov chain is $\eps$-close (in terms of total variation distance) to its stationary distribution $\mu$, regardless of $X^{(0)}$. The total variation distance between two distributions $\mu,\nu$ over $\Omega$ is defined as
$d_{\text{TV}}(\mu,\nu)=\sum_{\sigma \in \Omega} \frac{1}{2} \left|\mu(\sigma)-\nu(\sigma)\right|$.

\paragraph{Path Coupling}
The \emph{Path Coupling Lemma} by Bubley and Dyer \cite[Theorem 1]{bubley1997path} (also see \cite[Lemma 4.3]{feng2017can}) gives rise to a particularly easy way of designing couplings. In a simplified version, it says that it is enough to define the coupling of a Markov chain only for pairs of colorings that are adjacent, that is, differ at exactly one node. The expected number of differing nodes after one coupling step then can be used to bound the mixing time of the Markov chain.  
\begin{lemma}[Path Coupling \cite{bubley1997path}, simplified]\label[lemma]{pathcoupling}
For $\sigma,\sigma' \in [q]^V$,  let $\phi(\sigma,\sigma'):=|\{v \in V \colon \sigma_v \neq \sigma'_v\}|$.
If there is a coupling $(X,Y)\rightarrow (X',Y')$ of the Markov chain, defined only for $(X,Y)$ with $\phi(X,Y)=1$, that satisfies $\mathbb{E}[\phi(X',Y') \mid X,Y]\leq 1- \delta$ for some $0< \delta < 1$, then $\tau_{\text{mix}}(\eps)=O\left(\frac{1}{\delta}\cdot\log\left(\frac{n}{\eps}\right)\right)$.  
\end{lemma}

\section{Local Glauber Dynamics}\label[section]{sec:MarkovChain}

\paragraph{Local Glauber Dynamics}\label[section]{sec:LocalGlauber}
We define a transition from $X=(X_v)_{v\in V}$ to $X'=(X_v')_{v \in V}$ in one round as follows. 
Every node $v \in V$ marks itself independently with probability $0<\gamma < 1$. If it is marked, it proposes a new color $c_v \in [q]$ uniformly at random, independently from all the other nodes. If this proposed color does not lead to a conflict with the current and the proposed colors of any neighbor, that is, $c_v \notin \bigcup_{u \in N(v)} \{X_u, c_u\}$ and $c_u \notin \{X_v, c_v\}$ for any $u \in N(v)$\footnote{To simplify notation, we assume that $c_u=X_u$ in case $u$ is not marked.}, then $v$ accepts color $c_v$, thus sets $X'_v=c_v$. Otherwise, $v$ keeps its current color, that is, sets $X'_v=X_v$. Note that the condition $c_v \notin \bigcup_{u \in N(v)} \{X_u, c_u\}$ is necessary to guarantee reversibility of the Markov chain.

\paragraph{Stationary Distribution}
The local Glauber dynamics is ergodic: it is aperiodic, as there is always a positive probability of not changing any of the colors, and irreducible, since any (proper) coloring can be reached from any coloring.
Moreover, the chain might possibly start from an improper coloring, but it will never move from a proper to an improper coloring, that is, it is absorbing to proper colorings. 
It is easy to verify that this local Glauber dynamics, due to its symmetric update rule, satisfies the detailed balance equation for the uniform distribution, meaning that the transition from $X$ to $X'$ has the same probability as a transition from $X'$ to $X$ for proper colorings. The chain thus is reversible and has the uniform distribution over all proper colorings as unique stationary distribution. 

\paragraph{Mixing Time} Informally speaking, the Path Coupling Lemma says that if for all $X$ and $Y$ which differ in one node, we can define a coupling $(X,Y) \rightarrow (X',Y')$ in such a way that the expected number of nodes at which $X'$ and $Y'$ differ is bounded away from 1 from above, then the chain converges quickly. 
In \Cref{OSCoupling}, we formally describe such a path coupling, in \Cref{Properties}, we list necessary (but not necessarily sufficient) conditions for a node to have two different colors after one coupling step, which is then used in \Cref{EDiff} to bound the expected number of differing nodes by $1-\delta$ for some constant $0<\delta<1$, depending on $\alpha$. 
Application of \Cref{pathcoupling} then 
concludes the proof of \Cref{thm:col}.

\subsection{Description of Path Coupling}\label[section]{OSCoupling}

We look at two colorings $X$ and $Y$ that differ at a node $v_0\in V$ only. That is, $r=X_{v_0}\neq Y_{v_0}=b$, for some $r\neq g\in [q]$, which we will naturally refer to as red and blue, respectively, and $X_{v}=Y_{v}$ for all $v\neq v_0 \in V$. 
In the following, we explain how every node $v\in V$ comes up with a pair $(c_v^X,c_v^Y)$ of new proposals, which then will be accepted or rejected based on the local Glauber dynamics rules. 

\paragraph{Marking} In both chains, every node $v \in V$ is marked independently with probability $\gamma$, using the same randomness in both chains. In the following, we restrict our attention to marked nodes only; non-marked nodes are thought of proposing their current color as new color, i.e., $c_v^X=X_v$ and $c_v^Y=Y_v$.

\paragraph{Consistent, Mirrored, and Flipped Proposals} We introduce two possible ways of how proposals for a node $v$ can be sampled: \emph{consistently} and \emph{mirroredly}. For the consistent proposals, both chains propose the same randomly chosen color, that is, $c_{v}^X=c_{v}^Y=c$ for a u.a.r. $c \in [q]$. For the mirrored proposals, both chains assign the same random proposal if it is neither red nor blue, and a \emph{flipped} proposal (i.e., red to one and blue to the other chain) otherwise. More formally, $c_{v}^X=c$ and $c_{v}^Y=\overline{c}$  if $c\in \{r,b\}$ and $\overline{c}$ the element in $\{r,b\}\setminus\{c\}$, and $c_{v}^X=c_{v}^Y=c$ if $c\notin \{r,b\}$, for a u.a.r. $c \in [q]$. We say that $v$ has \emph{flipped} proposals if $c_{v}^X\neq c_{v}^Y$. Note that we say mirrored proposal to refer to the process of sampling mirroredly, and we say flipped if, as a result of sampling mirroredly, a node proposes different colors in the two chains.

\paragraph{Breadth-First Assignment of Proposals}
Let $B=\left\{ v \in V \setminus \{v_0\} \colon X_v \in \{r, b\}\right\} \subseteq V\setminus \{v_0\}$ be the set of nodes $v\neq v_0$ with current color red or blue, as well as $K=\left(\bigcup_{v \in B} N^+(v)\right)\setminus\{v_0\}$ its inclusive neighborhood, without $v_0$, where $N^+(v):=N(v)\cup\{v\}$. 
We ignore this set $K$ for the moment, and focus on the set $S\subseteq V \setminus K$ of marked nodes that are not adjacent to a node with color red or blue  (except for possibly $v_0$). Informally speaking, we will go through these nodes in a breadth-first manner, with increasing distance $d \geq 0$ to node $v_0$, and fix their proposals layer by layer, but defer the assignment of nodes not (yet) adjacent to a node with flipped proposals, as follows. We repeatedly add all (still remaining) nodes that have a node in the last layer with flipped proposals to a new layer, and sample their proposals mirroredly, thus perform a breadth-first assignment on nodes with flipped proposals only. All remaining nodes sample their proposals consistently. Note that this in particular guarantees that a node is sampled consistently only if it not adjacent to a node with flipped proposals. 

More formally, this can be described as follows. We define $M^0=F^0=\{v_0\}$, even if $v_0$ is not marked. 
For node $v_0$, if marked, the proposals are sampled consistently. For $d \geq 1$ and $v \in M^d$, the proposals are sampled mirroredly. For the subsequent layer, we restrict the attention to (new) neighbors of nodes in $M^d$ with flipped proposals only, i.e., consider $M^{d+1}=N\left(F^d\right)\setminus \bigcup_{i=0}^d M^d$ for $F^d=\{v \in M^d \colon c_{v}^X \neq c_{v}^Y\}$.

For all remaining (marked) nodes, that is, nodes in $S\setminus M$ and nodes in $K$, proposals are sampled consistently.
See \Cref{figg1} for an illustration of this breadth-first-based approach.  

\paragraph{Accept Proposals}
The proposals $(c_v^X)_{v \in V}$ and $(c_v^Y)_{v \in V}$ in the chains $X$ and $Y$ are accepted or rejected based on the local Glauber dynamics rules, leading to colorings $X',Y'\in [q]^V$. 

\begin{figure}
		\begin{center}
			\includegraphics[width=0.7\textwidth]{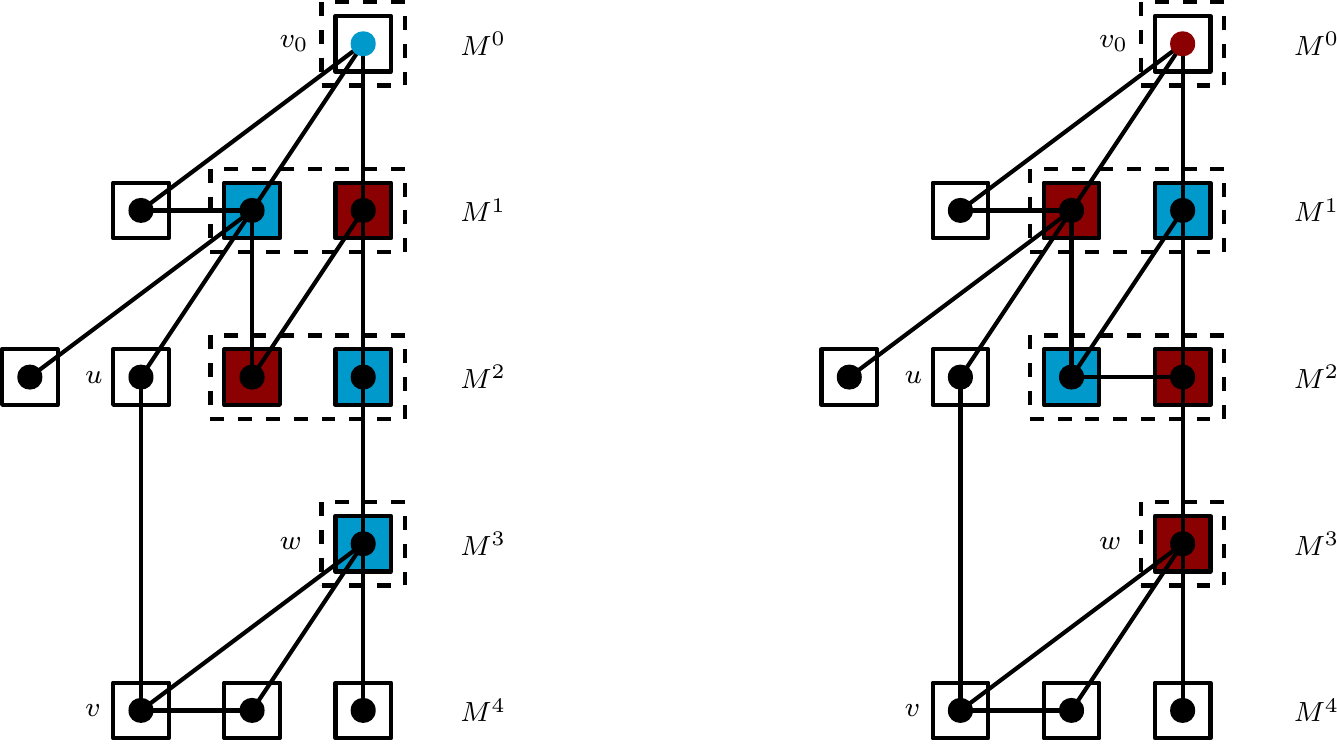}
		\end{center}
	\caption{The breadth-first layers $M^d$ for $d\geq 0$ of two chains that differ at $v_0\in M^0$. The disk color corresponds to the node's current color, where black means any color except red and blue. The color of the box around a node shows this node's proposed color, where white stands for any color (possibly also red or blue, but consistent). Dashed boxes indicate the sets $F^d$ of nodes with flipped proposals. Note that node $v$ appears in layer 4 even though it has distance 3 to $v_0$. This is because we perform the breadth-first assignment only on nodes with flipped proposals. $v$'s neighbor $u$ does not have flipped proposals, thus is in $M^2\setminus F^2$, which means that $u$'s neighbors are not added to the next layer. Only $v$'s neighbor $w\in F^3$ leads to $v$ being added to $M^4$.  
	}\label{figg1}
\end{figure}

\subsection{Properties of the Coupling}\label[section]{Properties}
The main observation is the following. If we ignore nodes with current colors red and blue for the moment, one can argue that $X'$ and $Y'$ can only differ at a node different from $v_0$ if its proposals are flipped. Flipped proposals, however, can only arise when the proposals are sampled mirroredly, which happens only if there is a node in the preceding layer with flipped proposals (due to the breadth-first order in which we assign the proposals). A node thus can lead to an inconsistency only if there is path in $G$ from $v_0$ to this node consisting of nodes with flipped proposals, called a \emph{flip path}. 

We will next make this intuition with the flip paths more precise, in two parts: for nodes in $S$ (that sample their proposals mirroredly if adjacent to a node with flipped proposals) in \Cref{FlipPath} and for nodes in $K$ (that always sample their proposals consistently) in \Cref{FlipPath2}. See \Cref{fig2} for an illustration of these two cases.

\begin{lemma}\label[lemma]{FlipPath}
If $X'$ and $Y'$ differ at $v\neq v_0 \in S$, there is a flip path $(v_0, \dotsc, v_{\ell}=v)\in F^0\times \dotsb \times F^{\ell}$ of length $\ell\geq 1$ in $G$, with the additional property that the proposal of $v$ is the opposite of the last color (in red and blue) seen on this path, in both chains. More formally, $c_Y=c_v^X\neq c_v^Y=c_X$, where $c_X=c_{v_{\ell-1}}^X$ and $c_Y=c_{v_{\ell-1}}^Y$ if $\ell>1$, and  $c_X=X_{v_0}$ and $c_Y=Y_{v_{0}}$ if $\ell=1$. 
\end{lemma}
\begin{proof}
We first argue that $v$'s proposals must be flipped and accepted in both chains. Trivially, acceptance of a consistent proposal in both chains or rejection in both chains leads to $X'_v=Y'_v$. 
Moreover, observe that flipped proposals are, by construction, either accepted in both or rejected in both chains, as flipping changes the role of red and blue, but not the overall behavior. 
Indeed, suppose, without loss of generality, that $c_v^X=c\in\{r,b\}$ is rejected by $X$. Thus, in particular, $v$ has a neighbor $u$ with current color or proposal $c$ in $X$. As we are restricting our attention to the set $S$ which does not have any adjacent node with current color red or blue, except for $v_0$, either $u=v_0$ or $u$ proposes $c$. So $u$ either must have different current colors (if $u=v_0$) or have mirrored proposals (if $v\in F^d$, then $u \in M^{d'}$ for some $d'\leq d+1$, because at the latest $v$'s flipped proposal leads to $u$ being added to the subsequent layer, by how we assign the proposals in breadth-first manner) and hence flipped proposals. Thus, $v$'s proposal $\overline{c}$ in $Y$ will be rejected by $Y$, since either $u=v_0\in N(v)$ has color $\overline{c}$ or $u\in N(v)$ proposes $\overline{c}$.

It thus remains to rule out the case of consistent proposals that are accepted in one and rejected in the other chain. 
Towards a contradiction, suppose that $v$ proposes the same color $c_v$ in both chains, and that it is accepted in one and rejected in the other. Since $X_v=Y_v$ and $c_v^X=c_v^Y$, this can happen only if $v$ is adjacent either to $v_0$ or to at least one node with flipped proposals, as otherwise all proposals and all current colors in $v$'s inclusive neighborhood would be the same, leading to the same behavior in both chains. In both cases, $v\in M^d$ for some $d\geq 1$, which means that its proposals are sampled mirroredly. Hence, $c_v \notin \{r,b\}$, as otherwise the proposals would be flipped. Now, since neither $v$'s current color nor $v$'s proposals is red or blue, and neighbors of $v$ can differ in their colors or proposals only if red or blue is involved, the proposals are either accepted or rejected in both chains. 
It follows that indeed only nodes in $S$ with flipped proposals that are accepted in both chains can have different colors in $X'$ and $Y'$. 

By construction of the layers, and since $v \in F^{\ell}$ for some $\ell \geq 1$, there must exist a sequence of nodes $v_1\in F^1, \dotsc, v_{\ell-1}\in F^{\ell-1}$ connecting $v_0$ to $v$ in $G$: a flip path of length $\ell$. Moreover, the proposal is accepted in a chain only if the proposed color is the opposite of the color (red or blue) that is seen on the path (either as proposal if $\ell >1$, or as current color of $v_0$ if $\ell=1$).
\end{proof}
\newpage

\begin{lemma}\label[lemma]{FlipPath2}
If $X'$ and $Y'$ differ at $v\neq v_0 \in K$, there is a path $(v_0, \dotsc, v_{\ell}=v)\in F^0\times \dotsb \times F^{\ell-1}\times  K$ of length $\ell\geq 1$ in $G$, called almost flip path, with the additional property that the proposal of $v$ is either red or blue, that is, $c_v=c_v^X=c_v^Y \in \{r,b\}$.
\end{lemma}
\begin{proof}
Since, by definition of the coupling, $v\in K$ samples its proposals consistently, $X'$ and $Y'$ can only differ at $v\neq v_0$ if the proposal is accepted in one and rejected in the other chain.
This can happen only if $v$ is adjacent to either $v_0$ or to at least one node with flipped proposals. Otherwise, all proposals and all current colors in $v$'s inclusive neighborhood would be the same, leading to the same behavior. Hence, $v$ is adjacent to some $u\in F^d$ for some $d \geq 0$. By construction of the layers, there must exist a sequence of nodes $v_1\in F^1, \dotsc, v_{\ell-1}=u\in F^{\ell-1}$ connecting $v_0$ to $v$ in $G$: an almost flip path of length $\ell=d+1$. Note that, in particular, because neighbors of nodes in $B$ are by definition sampled consistently (as they are in $K$), and a node at the end of an almost flip path has a neighbor with flipped proposals, this last node on an almost flip path must be in $K\setminus B$. 

The proposal $c_v$ is accepted in one and rejected in the other chain only if $c_v \in \{r,b\}$. In that case, the chain with the same color on the end of the path will reject, the other will (possibly) accept. 
\end{proof}

\begin{figure}[ht]
		\begin{center}
			\includegraphics[width=0.9\textwidth]{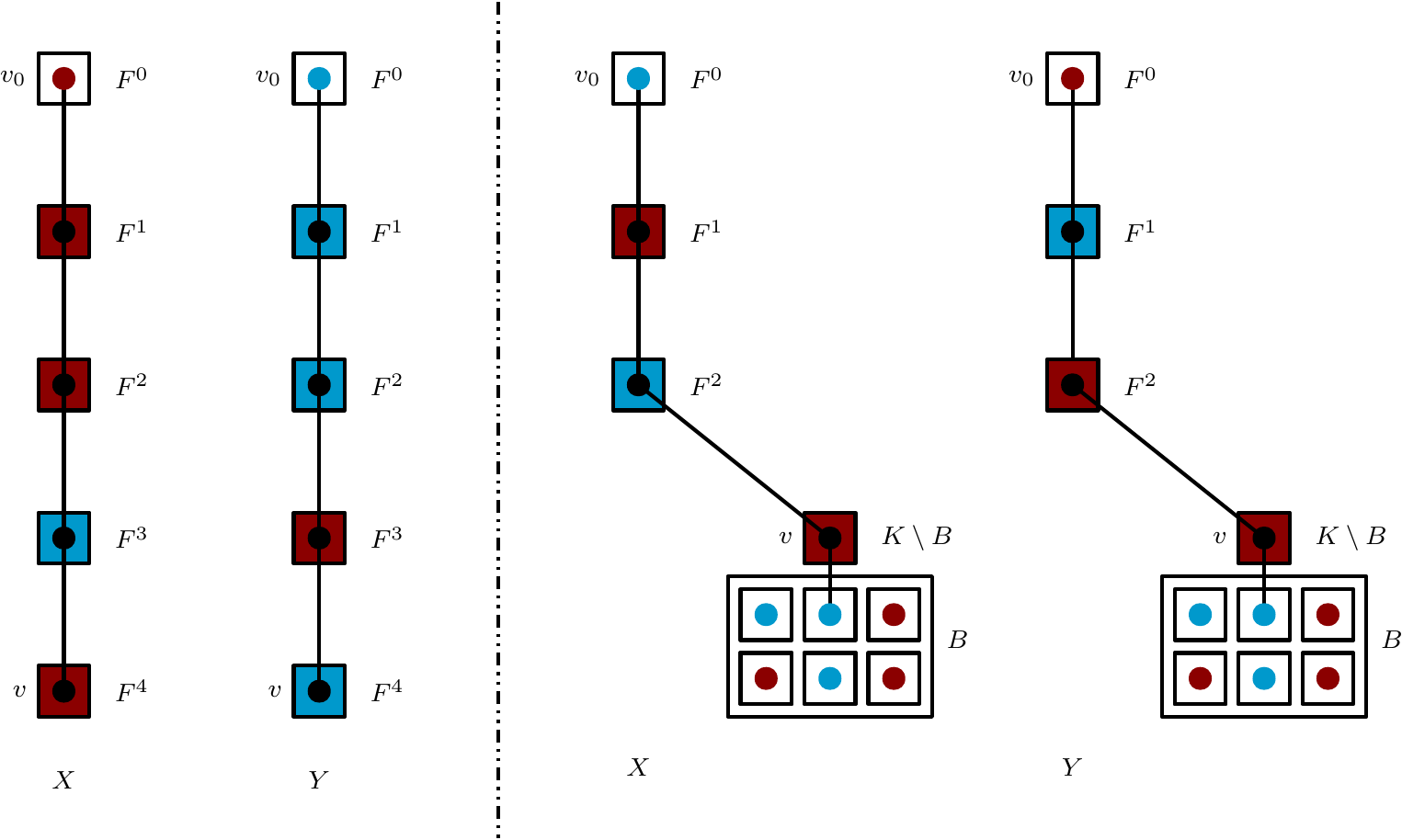}
		\end{center}
	\caption{A flip path on the left: $v$'s flipped proposals are accepted in both chains, yielding $X'_v=r$ and $Y'_v=b$.
	\\ An almost flip path on the right: $v\in K\setminus B$ samples its proposals consistently. In chain $X$, the proposal $r$ will be accepted, in chain $Y$, it will be rejected, leading to $X'_v=r\neq Y'_v=Y_v$. 
	The disk color corresponds to the node's current color, where black means any color except red and blue. The color of the box around a node indicates this node's proposed color, where white means any color ( also red and blue, but consistent).
	}	
	\label{fig2}
\end{figure}

\subsection{Bounding the Expected Number of Differing Nodes}\label[section]{EDiff}
We show that $\mathbb{E}[\phi(X',Y')\mid X,Y] \leq 1-\delta$ for some $0<\delta < 1$, by bounding the expectations $\mathbb{E}[\sum_{v\neq v_0 \in V}1\left(X'_v\neq Y'_v\right)\mid X,Y]$ and $\mathbb{E}[1\left(X'_{v_0}\neq Y'_{v_0}\right)\mid X,Y]$ separately. We will see that, as $\delta \rightarrow 0$, both terms can be bounded by  $\approx \frac{1}{\alpha}$, leading to an expected number of roughly $\frac{2}{\alpha}$, which is strictly less than 1 for $\alpha>2$.

\paragraph{Nodes $v\neq v_0$} 
\Cref{Properties}, or more precisely, \Cref{FlipPath,FlipPath2}, show that the number of nodes (different from $v_0$) that have different colors in $X'$ and $Y'$ can be bounded by the number of (almost) flip paths with an additional property. We will next see that the expected number of such (almost) flip paths can be expressed as a geometric series summing over the depths of the layers.

There are at most $\Delta^{\ell}$ paths $(v_0, \dotsc,v_{\ell})$ of length $\ell$ in $G$. Moreover, each such path has probability $\left(2\gamma/q\right)^{\ell-1} \gamma/q$ of being a flip or almost flip path with the mentioned additional property, since all intermediate nodes $v_1, \dotsc, v_{\ell-1}$ need to mark themselves and to propose one arbitrary color in $\{r,b\}$, and $v_{\ell}$ needs to mark itself and to propose the one color in $\{r,b\}$ as specified in \Cref{FlipPath,FlipPath2}, respectively. Note that a path in $G$ can either be a flip path or an almost flip path, but never both. Moreover, observe that node $v_0$ does not need to be marked.  
We get 
\begin{equation}\label{paths}\mathbb{E}\left[\sum_{v \neq v_0\in V} 1(X'_v\neq Y'_v) ~ \middle | ~X,Y\right]\leq \sum_{\ell=1}^{\infty}\Delta^{\ell}\cdot \left(\frac{2 \gamma}{q}\right)^{\ell-1}\cdot \frac{\gamma}{q} = \frac{1}{2}\sum_{\ell=1}^{\infty}  \left(\frac{2 \gamma \Delta}{q}\right)^{\ell} \leq \frac{\frac{\gamma \Delta}{q}}{1-\frac{2\gamma \Delta}{q}}.\end{equation}

\paragraph{Node $v_0$} Chains $X'$ and $Y'$ can agree at node $v_0$ only if at least one the proposals is accepted. For that, $v_0$ needs to be marked and its proposal $c_{v_0}=c_{v_0}^X=c_{v_0}^Y$ needs to be different from all the at most $\Delta$ current colors of its neighbors, that is, $c_v\notin \bigcup_{v\in N(v_0)}\{X_v\}$, which happens with probability at least $\gamma\left(1-\Delta/q\right)$. Moreover, the proposals of $v_0$'s neighbors (if marked) need to avoid at most three colors in $\{c_{v_0},r,b\}$, possibly less, which happens with probability at least $1-3\gamma/q$.
We thus get \begin{equation}\label{v0only}\mathbb{E}\left[1\left(X'_{v_0}\neq Y'_{v_0}\right)\right]\leq 1-\gamma\left(1-\frac{\Delta}{q}\right)\left(1-\frac{3\gamma}{q}\right)^{\Delta}.
\end{equation} 

\paragraph{Wrap-Up}
Overall, combining \Cref{paths,v0only}, we get
\begin{equation*}\begin{aligned}\mathbb{E}[\phi(X',Y') \mid X,Y]& 
\leq 1-\gamma\left(1-\frac{1}{\alpha}\right)e^{-\frac{6\gamma}{\alpha}} + \frac{\frac{\gamma}{\alpha}}{1-\frac{2\gamma}{\alpha}} 
=1 -\gamma e^{-\frac{6\gamma}{\alpha}}\left( 1-\frac{1}{\alpha}\left(1+\frac{e^{\frac{6\gamma}{\alpha}}}{1-\frac{2\gamma}{\alpha}}\right)\right).
\end{aligned}\end{equation*}
For $\alpha > 2$ and $\gamma:=\gamma(\alpha)$ small enough, this is strictly bounded away from 1 from above, where the hidden constant depends on $\alpha$ (but not on $\Delta$ or $n$). 

\bibliographystyle{alpha}
\bibliography{ref}

\end{document}